\documentclass[12pt, a4paper]{article}
\usepackage[utf8]{inputenc}
	 \usepackage{amsmath,amsthm}
	 \usepackage{amsfonts, amssymb}
 \usepackage{float}
 \usepackage{url}
\usepackage{graphicx}
 \usepackage{xcolor}
\usepackage{enumerate}
\usepackage{enumitem}
\usepackage{array}
\usepackage[urlcolor=black,linkcolor=black,citecolor=black,colorlinks=true]{hyperref}
\usepackage[left=2.9cm, right=2.9cm, top=2.25cm, bottom=3.25cm]{geometry}
\usepackage{pict2e}
\usepackage{placeins}
\usepackage{tikz}

\numberwithin{equation}{section}

\newtheorem{theoreme}{Th\'{e}or\`{e}me}[section]
\newtheorem{lemme}[theoreme]{Lemma}

\newtheorem{proposition}[theoreme]{Proposition}
\newtheorem{corollaire}[theoreme]{Corollary}

\newtheorem{remarque}[theoreme]{Remark}
\newtheorem{hypothese}[theoreme]{Assumption}

\title{Continuous-time modeling and bootstrap for chain-ladder reserving}

\author{Nicolas Baradel\footnote{Inria, CMAP, CNRS, \'{E}cole polytechnique, Institut Polytechnique de Paris, 91200 Palaiseau, nicolas.baradel@polytechnique.edu.}}

\begin{document}

\maketitle	

\begin{abstract}

We revisit the famous Mack's model which gives an estimate for the conditional mean squared error of prediction of the chain-ladder claims reserves. We introduce a stochastic differential equation driven by a Brownian motion to model the accumulated total claims amount for the chain-ladder method. Within this continuous-time framework, we propose a bootstrap technique for estimating the distribution of claims reserves. It turns out that our approach leads to inherently capturing asymmetry and non-negativity, eliminating the necessity for additional assumptions. We conclude with a case study and comparative analysis against alternative methodologies based on Mack's model.

\end{abstract}
\section{Introduction}

The chain-ladder technique is a cornerstone of reserving in non-life insurance. Several stochastic frameworks have been developed to derive the chain-ladder reserve estimator. One notable approach employs maximum likelihood estimation within an over-dispersed Poisson generalized linear model, yielding equivalent reserve estimates \cite{kremer1985einfuhrung}, \cite{mack1991simple}, and \cite{kuang2009chain}.

\medbreak

Mack's model \cite{mack1993distribution} provides a distribution-free approach to obtain estimators and claims reserves similar to the chain-ladder method, relying on minimal assumptions about conditional moments. This framework introduces a stochastic model that also enables estimation of the conditional Mean Squared Error of Prediction (MSEP), a measure of prediction uncertainty.

\medbreak

Several studies have explored Mack's model, some adopting stronger assumptions aligned with Mack's framework. For example, \cite{buchwalder2006mean} introduced a time series for claims development with independent and identically distributed noise satisfying Mack's assumptions. However, this approach may produce negative values for incurred claims, an issue that cannot be resolved by conditioning the noise without compromising increment independence, as noted in \cite{mack2006mean}.

\medbreak

In this paper, we introduce a continuous-time model for claims development based on a well chosen stochastic differential equation driven by Brownian motion. We demonstrate that our model adheres to Mack's assumptions, and in a specific scenario, we can leverage all of Mack's estimators. The primary advantage lies in our ability to simulate total claims reserves using a parametric bootstrap method, which inherently incorporates asymmetry and non-negativity without the need for residual computation or additional assumptions.

\medbreak

Several studies have investigated continuous-time frameworks for the chain-ladder method. For instance, \cite{miranda2013continuous} proposes a continuous framework for loss reserving in non-life insurance, utilizing a bivariate density approach with a kernel-based method to model individual claims within the development triangle, aiming to estimate the distribution in the lower triangle. \cite{bischofberger2020continuous} uses a continous-time approach with marked point processes to capture the timing and magnitude of individual claim payments. In contrast, our work focuses on aggregated data, introducing a continuous-time stochastic process with continuous paths to model the evolution of aggregated incurred claims, satisfying Mack's assumptions.

\medbreak

The paper is organized as follows. Section 2 presents Mack's general model, with a review of key estimators. Section 3 introduces our continuous-time model for the accumulated total claims amount, for which we derive several properties and establishing its connection to Mack's model. Section 4 describes a bootstrap procedure tailored to the continuous-time model, addressing uncertainty in parameter estimation. Finally, Section 5 provides a case study that assesses the impact of the continuous-time framework and compares it with alternative approaches based on Mack's model.

\section{Mack's model}

Mack's model provides a probabilistic framework that aligns with the chain-ladder method. It calculates the conditional MSEP for reserves without requiring a specific distribution: it imposes constraints on the conditional moments of the underlying process.

\medbreak

The model introduces the process $(C_{i,j})_{1 \leq i, j \leq n}$ which represents the accumulated total claims amount for both occurrence year $i$ and development year $j$ across $n$ years of observations. For each $1 \leq k \leq n$, we define:
    \[
        \mathcal{F}^i_k := \sigma\left(C_{i, j}, j \leq k\right), \ \ 1 \leq i \leq n.
    \]
We make the following assumption:
    \begin{hypothese}\label{H_mack} 
        \leavevmode\begin{itemize}
            \item[H1] The random variables $(C_{i_{1}, j})_{1 \leq j \leq n}$ and $(C_{i_{2}, j})_{1 \leq j \leq n}$ are independent for $i_{1} \not= i_{2}$. 
            \item[H2] For $1 \leq j \leq n - 1$, there exists $F_{j} > 0$ such that
                \[
                    \mathbb{E}(C_{i, j+1} \mid \mathcal{F}_{j}^i) = F_{j}C_{i,j}, \ \ \ 1 \leq i \leq n.
                \]
            \item[H3] For $1 \leq j \leq n - 1$, there exists $\Sigma_{j} \geq 0$ such that
                \[
                    Var(C_{i, j+1} \mid \mathcal{F}_{j}^i) = \Sigma_{j}^{2}C_{i,j}, \ \ \ 1 \leq i \leq n.
                \]
        \end{itemize}
    \end{hypothese}
From the above assumption, we can derive the general expressions for the first two moments across all dates:

\begin{lemme}\label{clcor}
For all  $1 \leq i \leq n$ and $s \leq j \leq n$,
\[
        \begin{aligned}
            \mathbb{E}(C_{i,j} \mid \mathcal{F}_{s}^i) &= \left(\prod_{k = s}^{j-1}F_{k}\right)C_{i,s},\\
            Var(C_{i,j} \mid \mathcal{F}_{s}^i) &= \left(\sum_{k=s}^{j-1}\left[\left(\prod_{\ell=k+1}^{j-1}F_\ell^2\right)\Sigma_{k}^{2}\left(\prod_{\ell=s}^{k-1}F_\ell\right)\right]\right)C_{i, s}.
        \end{aligned}
        \]
    \end{lemme}
\noindent Mack provides accurate estimators for both the $F$'s and the $\Sigma^2$'s:

\begin{equation}\label{est_fs}
    \begin{aligned}
    \widehat{F}_{j} &:= \frac{\sum_{i = 1}^{n - j}C_{i, j+1}}{\sum_{i = 1}^{n - j}C_{i,j}},  &1 \leq j \leq n-1,\\
    \widehat{\Sigma}_{j}^2 &:= \frac{1}{n-j-1}\sum_{i = 1}^{n-j}C_{i,j}\left(\frac{C_{i,j+1}}{C_{i,j}} - \widehat{F}_{j}\right)^{2}, &1 \leq j \leq n-2.
    \end{aligned}
\end{equation}
Several methods exist for estimating $\Sigma_{n-1}^2$, as discussed in \cite{mack1993distribution}. Mack also provides an unbiased estimator for the ultimate value:
\begin{equation*}
	\widehat{C}_{i,n} := C_{i, n-i+1}\left(\prod_{k=n-i+1}^{n-1}\widehat{F}_{j}\right), \quad 2 \leq i \leq n,
\end{equation*}
which consequently leads to the reserve estimator:
\begin{equation*}
	\widehat{R} := \sum_{i=2}^{n} \widehat{C}_{i,n} - C_{i, n-i+1}.
\end{equation*}

Moreover, he presents an estimator for the conditional MSEP of the reserves, accounting for uncertainty arising from parameter estimation. An alternative method to assess the conditional MSEP involves employing a bootstrap approach. For a comprehensive introduction to this technique in the realm of insurance reserving, refer to \cite{england2006predictive}. Unlike solely estimating the conditional MSEP of the reserve, bootstrap analysis offers insight into the entire distribution.

\medbreak

The aim of this paper is to establish a continuous-time framework using stochastic differential equations that fulfill Assumption \ref{H_mack}.

\medbreak

In \cite{buchwalder2006mean}, a time series methodology was employed, yielding the following model:

\begin{equation}\label{Cij_ts}
    C_{i,j+1} = F_{j} C_{i,j} + \Sigma_{j} \sqrt{C_{i,j}} \varepsilon_{i,j},
\end{equation}
where the $\varepsilon$'s represent independent variables with a mean of zero and a variance of one. A key limitation of the formulation above is that sampling from \eqref{Cij_ts} may yield negative values for the $C$'s terms, which is inconsistent. While the authors propose using a conditional distribution for $\varepsilon_{i,j}$ based on $C_{i,j-1}$,  this approach, as noted in \cite{mack2006mean}, compromises the independence assumption for the $\varepsilon$'s terms, rendering it unsuitable.

Our framework, detailed in the following section, extends the yearly-based model in \eqref{Cij_ts} to a continuous-time setting, effectively addressing the issue of negative values.

\section{A continuous-time model}

Let $\Omega_W := C([1, 2n], \mathbb{R}^n)$ denote the space of continuous functions mapping $[1, 2n]$ to $\mathbb{R}^n$, where functions start with value 0 at 1. We denote by $W(\omega) = \omega$ the canonical process and let $\mathbb{P}_W$ be the Wiener measure defined on the Borel sets of $\Omega_W$. Consequently, $W = (W^i)_{1 \leq i \leq n}$ comprises $n$ independent Brownian motions. Let $\Omega_1$ be a Polish space and $\mathbb{P}_{1}$ a Borel measure on $\Omega_1$. Finally, we define $\Omega := \Omega_1 \times \Omega_W$ and the product measure $\mathbb{P} := \mathbb{P}_{1} \otimes \mathbb{P}_{W}$ on the Borel sets of $\Omega$.

\medbreak

We introduce the following filtrations, which represent the knowledge at development time $t$ for an occurrence year $i$:
    \[
        \mathcal{F}_{t}^i := \sigma(C_{1}^{i}; \ W_{s}^i, \ s \leq t), \quad 1 \leq i \leq n, \ 1 \leq t \leq n,
    \]
in which the $(C_{1}^{i})_{1 \leq i \leq n}$ are random variables defined on $\Omega_{1}$ and valued in $\mathbb{R}_{+}$. We define the filtration of the entire knowledge at time $t \in [1, n]$.
    \[
        \mathcal{F}_{t} := \sigma(C_{1}^{i}, i \leq t; \ W_{s}^i, \ i+s \leq t+1, s \leq n), \quad 1 \leq t \leq 2n.
    \]
Hereafter, all random variables are considered within the probability space $(\Omega, \mathcal{F}_{2n})$. Let $(C^i_t)_{t \in [1, n]}^{i \in \{1, \ldots, n\}}$ represent the processes of accumulated total claims amount for occurrence year $i$ at development date $t$.

\begin{hypothese}
        \leavevmode\begin{itemize}
            \item[H1'] The random variables $(C_{1}^{i})_{1 \leq i \leq n}$ are square integrable and independent.
	\item[H2'] There exist two measurable and bounded functions $f : [1, n] \rightarrow \mathbb{R}$ and $\sigma : [1, n] \rightarrow \mathbb{R}_{+}$ such that, for all $1 \leq i \leq n$, $(C_{t}^i)_{1 \leq t \leq n}$ is the unique strong solution of the stochastic differential equation:
\begin{equation} \label{EDS_C}  
    C_{t}^i = C_{s}^i + \int_{s}^{t} f_{u} C_u^idu + \int_{s}^{t} \sigma_{u} \sqrt{C_u^i}dW_u^i, \quad 1 \leq s \leq t \leq n.
\end{equation}
\end{itemize}
\end{hypothese}

The processes $C^i$ are well-defined by \eqref{EDS_C} since these stochastic differential equations possess a unique (non-negative) strong solution, as established in, for instance, \cite{yamada1971uniqueness} or \cite[Theorem 4.6.11]{meleard2016modeles}. Furthermore, they satisfy:
\begin{equation}\label{EDS_L2}
	\mathbb{E}\left[\sup_{1 \leq t \leq n} (C_{t}^i)^2\right] < +\infty, \quad 1 \leq i \leq n.
\end{equation}

\begin{remarque}\label{feller}
In the specific case where the coefficients $f$ and $\sigma$ are constant, this process is referred to as the Feller diffusion, originally introduced in \cite{feller1971introduction}, see for example: \cite{revuz2013continuous}, \cite{ikeda2014stochastic}, and \cite{meleard2016modeles}.
\end{remarque}

The above process bears resemblance to the \emph{Cox-Ingersoll-Ross} process commonly employed in finance, yet it distinguishes itself by lacking mean reversion. Notably, it is well-known in population dynamics studies, as it can be interpreted as the limit of the \emph{Galton-Watson branching process}. Its primary characteristic is the \emph{branching property}. This property is also satisfied in the classical Mack chain-ladder model and we find it again in the continuous-time model in a general form.

\begin{lemme}\label{branching}
The processes $(C_{t}^i)_{1 \leq t \leq n}$ satisfy the branching property: if for $1 \leq i \leq n$, $(C_{t}'^i)_{1 \leq t \leq n}$ is another independent process satisfying \eqref{EDS_C} with a different Brownian motion, then $(C_{t}^i+C_{t}'^i)_{1 \leq t \leq n}$ also satisfies \eqref{EDS_C} with yet another Brownian motion.
\end{lemme}
\begin{proof}
This property is standard when considering constant coefficients. For instance, refer to \cite[Proposition 4.7.1]{meleard2016modeles} or, for a more general approach, refer to \cite{lamperti1967continuous}. With bounded time-dependent coefficients, the proof remains straightforward, without any significant differences.
\end{proof}

\begin{remarque}
The branching property of Lemma \ref{branching} above implies the following consequence: if we consider a portfolio consisting of $d$ independent components, each governed by the dynamics defined in (\ref{EDS_C}) with identical parameters $f$ and $\sigma$, then the aggregation of these $d$ components will also exhibit the dynamics described by (\ref{EDS_C}). Consequently, it will yield the same aggregated reserve distribution. Similarly, dividing a portfolio into two homogeneous independent sub-portfolios maintains the same dynamics and, consequently, the same aggregated reserve distribution. Implicit in this assertion is the assumption that all constituents of a portfolio are independent.
\end{remarque}

\begin{remarque}
We began defining the process at $t = 1$, with the initial condition $(C_{1}^i)_{1 \leq i \leq n}$ as random variable. This approach aligns with the Mack's general framework, as we make no assumptions about $(C_{1}^i)$ other than its implicit squared integrability. Additionally, extending the process $(C_{t}^i)_{1 \leq t \leq n}$ defined in \eqref{EDS_C} back to $t=0$ would require $C_{0}^{i} > 0$, which is not relevant. Implicitly, the randomness of $(C_{1}^i)_{1 \leq i \leq n}$, corresponding to the year of occurrence, follows a different process. This process does not need to be defined for the chain-ladder technique to derive the reserves and their conditional MSEP or distribution, conditional on the current information. However, it should be defined in order to simulate $C_{n}^{n+1}$.
\end{remarque}

We now derive the first two conditional moments of the $C$'s to verify Assumption \ref{H_mack}.

\begin{proposition}
The first two conditional moments of the processes $(C_{t}^i)_{1 \leq t \leq n}$ are, for all $1 \leq s \leq t \leq n$ and $1 \leq i \leq n$,
\[
    \begin{aligned}
        \mathbb{E}\left(C_{t}^i \mid \mathcal{F}_s^i\right) &= C_{s}^i e^{\int_{s}^{t} f_u du },\\
        Var\left(C_{t}^i \mid \mathcal{F}_s^i\right) &= C_s^i\int_{s}^{t}\sigma_u^2e^{\int_{s}^u f_z dz + \int_{u}^t 2f_z dz}du.
    \end{aligned}
\]
\end{proposition}
\begin{proof}
Fix $i \in \{1, \ldots, n\}$. For convenience, denote $C^i$ as $C$, $\mathcal{F}^i$ for $\mathcal{F}$, and $W^i$ as $W$ throughout this proof.\\
\textbf{1.} Applying the expected value operator $\mathbb{E}$ to (\ref{EDS_C}) and utilizing (\ref{EDS_L2}) for the local martingale yields:
\[
    \mathbb{E}\left(C_{t} \mid \mathcal{F}_s\right) = C_{s} + \int_{s}^{t} f_{u} \mathbb{E}\left(C_{u} \mid \mathcal{F}_s\right)du. 
\]
This forms a simple linear homogeneous ordinary differential equation with the unique solution:
\begin{equation}\label{esperance}
    \mathbb{E}\left(C_{t} \mid \mathcal{F}_s\right) = C_{s} e^{\int_{s}^{t} f_u du }.
\end{equation}
\textbf{2.} It\^{o}'s formula gives:
\[\begin{aligned}
    C_{t}^2 &= C_{s}^2 + 2\int_{s}^{t}  f_u C_u^2du + 2\int_{s}^{t}  \sigma_u C_u\sqrt{C_u}dW_u + \int_{s}^{t} \sigma_u^2 C_udu.
    \end{aligned}
\]
Introducing the stopping times $T_{m} := \inf\{t \geq s: C_{t} = m\}$, which tends to infinity a.s. as $m \rightarrow +\infty$, and considering the process $C$ on $\mathbb{R}_{+}$, we apply the expected value operator:
\begin{equation}
        \mathbb{E}\left(C_{t\wedge T_m}^2 \mid \mathcal{F}_s\right) = C_{s}^2 + 2\mathbb{E}\left(\int_{s}^{t\wedge T_m}  f_u C_u^2 du \mid \mathcal{F}_s\right) + \mathbb{E}\left(\int_{s}^{t\wedge T_m} \sigma_u^2 C_u du \mid \mathcal{F}_s\right).
\end{equation}
Taking the limit as $m \rightarrow +\infty$, and using (\ref{EDS_L2}) along with the dominated convergence theorem, we obtain:
\begin{equation}\label{EC2}
        \mathbb{E}\left(C_{t}^2 \mid \mathcal{F}_s\right) = C_{s}^2 + 2\int_{s}^{t}  f_u \mathbb{E}\left(C_u^2 \mid \mathcal{F}_s\right) du + \int_{s}^{t} \sigma_u^2 \mathbb{E}\left(C_u \mid \mathcal{F}_s\right) du.
\end{equation}
From (\ref{esperance}), we have
\[
    \mathbb{E}\left(C_{t} \mid \mathcal{F}_s\right)^2 = C_{s}^2 e^{2\int_{s}^{t} f_u du }
\]
thus, $t \mapsto\mathbb{E}\left(C_{t} \mid \mathcal{F}_s\right)^2$ satisfies the following ordinary differential equation: 
\[
    \mathbb{E}\left(C_{t} \mid \mathcal{F}_s\right)^2 = C_{s}^2 + 2\int_{s}^{t} f_{u} \mathbb{E}\left(C_{u} \mid \mathcal{F}_s\right)^2du.
\]
Combining it with (\ref{EC2}) and (\ref{esperance}) gives:
\[
    Var\left(C_{t} \mid \mathcal{F}_s\right) =  2\int_{s}^{t}  f_u Var(C_u \mid \mathcal{F}_s)du + C_{s}\int_{s}^{t} \sigma_u^2  e^{\int_{s}^{u} f_z dz }du .
\]
It is a linear non-homogeneous ordinary differential equation whose unique solution is:
\[
    Var\left(C_{t} \mid \mathcal{F}_s\right) = C_s\int_{s}^{t}\sigma_u^2e^{\int_{s}^u f_z dz + \int_{u}^t 2f_z dz}du.
\]
\end{proof}

\begin{corollaire}\label{C_param}
The processes $(C_{t}^i)_{1 \leq t \leq n}$ for $1 \leq i \leq n$ satisfy assumptions \textit{H1}, \textit{H2} and \textit{H3} of Assumption \ref{H_mack} by setting, for $1 \leq j \leq n$:
\[
    \begin{aligned}
    F_j &:= e^{\int_{j}^{j+1}f_udu},\\
    \Sigma_j^2 &:= \int_{j}^{j+1}\sigma_u^2e^{\int_{j}^u f_z dz + \int_{u}^{j+1} 2f_z dz}du.
    \end{aligned}
\]
\end{corollaire}

\begin{remarque}
We directly obtain the expected values of the continuous process $C_t^i$ conditional on $\mathcal{F}_{s}^i$, which correspond to the discrete ones stated in Lemma \ref{clcor}.
\end{remarque}

There is no need to precisely track the $(f, \sigma) : t \rightarrow (f_t, \sigma_t)$ function continuously. To simplify matters, we introduce an additional assumption: that the function remains constant over each one-year interval. Consequently, we establish a connection between the estimators derived from the classic framework and our continuous-time framework.

    \begin{hypothese}\label{H_constant}
The functions $f$ and $\sigma$ are constant on each $[t, t+1)$, i.e., for $1 \leq t < n$:
    \[
        \begin{aligned}
            f_t &:= \sum_{j=1}^{n}f_j\mathbf{1}_{[j, j+1)}(t),\\
            \sigma_t &:= \sum_{j=1}^{n}\sigma_j\mathbf{1}_{[j, j+1)}(t).
        \end{aligned}
    \]
    \end{hypothese}

\begin{lemme}\label{smallf}
    Under the additionnal Assumption \ref{H_constant}, the relation in Corollary \ref{C_param} simplifies to
    \[
        \begin{aligned}
            F_j &= e^{f_j} \\
            \Sigma^2_j &= \frac{\sigma_j^2}{f_j}\left(e^{2f_j} - e^{f_j}\right)
        \end{aligned} \ \ \Longleftrightarrow
        \begin{aligned}
            f_j &= \log(F_j) \\
            \sigma^2_j &= \frac{\Sigma_j^2 \log(F_j)}{F_j(F_j-1)}
        \end{aligned}
    \]
\end{lemme}
\begin{proof}
The proof follows straightforwardly from computing the simple integrals.
\end{proof}

Note that in \eqref{EDS_C}, as $C_{t}^i$ approaches zero, both the term preceding $dt$ and the one preceding $dW_t^i$ vanish. We will now discuss the conditional distribution of $C_{t}^{i}$, particularly emphasizing that while it is possible for $C_{t}^{i}$ to reach zero, this occurrence is practically negligible.

\begin{lemme}\label{laplace}
Under Assumption \ref{H_constant}, for $j \leq t \leq j+1$ and $z > -\frac{2f_j}{\sigma_j^2\left(e^{f_{j}(t-j)} - 1\right)}$, the Laplace transform of $C_{t}^i$ conditional on $\mathcal{F}_{j}^i$ is given by:
	\[
		g_{i,j,t}(z) = \mathbb{E}\left(e^{-z C_{t}^i} \, \bigg| \, \mathcal{F}_{j}^i\right) =  \exp\left(-\frac{2f_{j}e^{f_{j}(t-j)}C_j^{i}z}{2f_j + \sigma_j^2\left(e^{f_{j}(t-j)} - 1\right)z}\right).
	\]
\end{lemme}
\begin{proof}
This result is standard; see, for example, \cite[Proposition 4.7.1]{meleard2016modeles} or \cite[Proposition 4.4]{dawson2017introductory}. Note that the latter contains a typographical error, omitting factors of 2 in the expression, though its proof remains correct.
\end{proof}

\begin{remarque}\label{PC0}
For $j \leq t \leq j+1$, under Assumption \ref{H_constant}, we have
	\[
		\mathbb{P}(C_{t}^i = 0 \mid \mathcal{F}_{j}^i) \  = \lim_{z \rightarrow +\infty}g_{i,j,t}(z) = \exp\left(-\frac{2f_{j}e^{f_{j}(t-j)}C_j^i}{\sigma_j^2\left(e^{f_{j}(t-j)} - 1\right)}\right).
	\]
This implies that the processes $(C_{t}^i)_{t \geq 1}$ can reach 0 (and remain there). However, in practice, as we will observe, this probability is often numerically close to 0, signifying the scenario where all claims ultimately cost 0. Moreover, the distribution of $C_{t}^{i}$, conditioned to be positive, is continuous.
\end{remarque}

In a Cox-Ingersoll-Ross framework, the conditional marginal distributions of the process follow a continuous distribution, specifically a non-central chi-squared distribution. Within our framework, these distributions assign positive probability to zero, yet they are fully characterized as described below.

\begin{lemme}\label{loiS}
	Let $N \sim \mathcal{P}(\lambda)$ with $\lambda >0$ and $(X_k)_{k \geq 1} \overset{i.i.d.}{\sim} \mathcal{E}(\beta)$ with $\beta > 0$. Define
	\[
	S := \sum_{k=1}^{N} X_k.
	\]
For $j \leq t \leq j+1$, if 
	\[
		\lambda = \frac{2f_{j}e^{f_{j}(t-j)}C_j^{i}}{\sigma_j^2\left(e^{f_{j}(t-j)} - 1\right)}, \quad \beta = \frac{2f_{j}}{\sigma_j^2\left(e^{f_{j}(t-j)} - 1\right)},
	\]
then the conditional distribution of $C_{t}^i$ given $\mathcal{F}_{j}^i$ is identical to that of $S$ given $\mathcal{F}_{j}^i$.
\end{lemme}
\begin{proof}
For $z > -\beta$, the Laplace transform of $S$ is
	\[
		g_{S}(z) := \mathbb{E}\left(e^{-z S}\right) = g_{N}(-\log(g_X(z)),
	\]
where $g_N(z) = e^{\lambda (e^{-z} - 1)}$ is the Laplace transform of $N$ and $g_X(z) = \frac{\beta}{\beta + z}$ is the Laplace transform of $X_{1}$. Thus,
	\[
		g_{S}(z) = e^{\frac{-\lambda z}{\beta + z}}.
	\]
Equating $g_{S}$ to $g_{i,j,t}$ and solving for $\lambda$ and $\beta$ yields the result.
\end{proof}
\begin{remarque}\label{loiSr}
The preceding lemma provides an exact method for simulations of $S$, eliminating the need for an Euler scheme with fine discretization. Since
	\[
		S \mid N \sim \mathcal{G}(N, \beta), 
	\]
where $\mathcal{G}$ denotes the Gamma distribution, one can simulate $S$ by first generating $N$ and then sampling from the conditional distribution $S \mid N$. This approach ensures both accuracy and computational efficiency in the simulation process.
\end{remarque}

This property is highly valuable for bootstrap simulations. Since Mack's assumptions are satisfied, we obtain the same estimators for the reserves and can compute the same conditional MSEP. Our goal is to propose a bootstrap methodology, tailored for our continuous-time framework, which will enable the estimation of the distribution of the reserves.

\section{The bootstrap methodology}\label{C_bootstrap}

There are the two classical steps:

\begin{enumerate}
    \item Bootstrapping the parameters: the $F$'s and the $\Sigma$'s, to account for the \emph{estimation error};
    \item Simulating the lower part of the triangle using the bootstrapped coefficients to incorporate the \emph{process error}.
\end{enumerate}
We adapt the bootstrap approach described in \cite{england2006predictive} to our framework.

\medbreak

\noindent \textbf{1.} For bootstrapping the coefficients, 
\[
    C_{j+1}^{i, m} = C_{j}^{i} + \int_{j}^{j+1}\widehat{f}_{j}C_u^{i, m}du + \int_{j}^{j+1}\widehat{\sigma}_{j}\sqrt{C_u^{i, m}}dW_u^{i, m}, \quad i + j \leq  n, \ 1 \leq m \leq M.
\]
Note that the above stochastic differential equation uses $C_{j}^{i}$ as its initial condition, not $C_{j}^{i, m}$. By Lemma \ref{loiS} and Remark \ref{loiSr}, the conditional simulation of $C_{j+1}^{i, m}$ is done with:
\[
	C_{j+1}^{i, m} \sim \mathcal{G}(N_{i,j}^m, \widehat{\beta}_{j}) \text{ with } N_{i,j}^m \sim \mathcal{P}(\widehat{\lambda}_{i,j}), \quad i + j \leq  n, \ 1 \leq m \leq M,
\]
where $\widehat{\beta}_{j} := \frac{2\widehat{f}_{j}}{\widehat{\sigma}_j^2\left(e^{\widehat{f}_{j}} - 1\right)}$ and $\widehat{\lambda}_{i,j} := \widehat{\beta}_{j}\,e^{\widehat{f}_{j}}C_j^{i}$.

\medbreak

We obtain the new estimators $(\widehat{F}_{j}^m)$ and $(\widehat{\Sigma}_{j}^m)$ defined as, for all $1 \leq m \leq M$:

\begin{equation}
    \begin{aligned}
    \widehat{F}_{j}^m &:= \frac{\sum_{i = 1}^{n - j}C_{j+1}^{i, m}}{\sum_{i = 1}^{n - j}C_{j}^i}, & 1 \leq j \leq n-1, \\
    (\widehat{\Sigma}_{j}^{m})^2 &:= \frac{1}{n-j-1}\sum_{i = 1}^{n-j}C_{j}^i\left(\frac{C_{j+1}^{i, m}}{C_{j}^i} - \widehat{F}_{j}^m\right)^{2}, & 1 \leq j \leq n-2,
    \end{aligned}
\end{equation}
and $\widehat{\Sigma}_{n-1}^m := \min(\frac{(\widehat{\Sigma}_{n-2}^m)^2}{\widehat{\Sigma}_{n-3}^m}, \widehat{\Sigma}_{n-3}^m, \widehat{\Sigma}_{n-2}^m)$ as proposed in \cite{mack1993distribution}. We then derive $(\widehat{f}_{j}^m)$ and $(\widehat{\sigma}_{j}^m)$ using Lemma \ref{smallf}.

\noindent \textbf{2.} For bootstrapping the process error: 
\[
    C_{n}^{i, m} = C_{n-i+1}^{i} + \int_{n-i+1}^{n}\widehat{f}_{u}^{m}C_u^{i, m}du + \int_{n-i+1}^{n}\widehat{\sigma}_{u}^{m}\sqrt{C_u^{i, m}}dW_u^{i, m}, \quad 2 \leq i \leq  n, \ 1 \leq m \leq M.
\]
Again, by Lemma \ref{loiS} and Remark \ref{loiSr}, the conditional simulation of $C_{j+1}^{i, m}$ is done with:
\[
	C_{j+1}^{i, m} \sim \mathcal{G}(N_{i,j}^m, \widehat{\beta}_{j}^m) \text{ with } N_{i,j}^m \sim \mathcal{P}(\widehat{\lambda}_{i,j}^m), \quad i + j >  n, \ 1 \leq m \leq M,
\]
where $\widehat{\beta}_{j}^m := \frac{2\widehat{f}_{j}^m}{(\widehat{\sigma}_j^m)^2\left(e^{\widehat{f}_{j}^m} - 1\right)}$ and $\widehat{\lambda}_{i,j}^m := \widehat{\beta}_{j}^m\,e^{\widehat{f}_{j}^m}C_j^{i, m}$, and where $C_{n-i+1}^{i, m} := C_{n-i+1}^{i}$.

\noindent \textbf{3.} It yields to the simulation of the reserves:
\begin{equation}\label{r_bootstrap}
	R^m := \sum_{i=2}^{n} C_{n}^{i, m} - C_{n-i+1}^{i}, \quad 1 \leq m \leq M.
\end{equation}
The vector $(R^m)_{1 \leq m \leq M}$ approximates the distribution of the reserves, conditional on our observations.

\begin{remarque}
We described a bootstrap procedure to simulate the reserves. This method can be adapted to simulate $C_{n}^{n+1}$. Given that the only assumption on $C_{1}^{n+1}$ is its square integrability, an additional assumption is needed to simulate it. One approach is to use the exposure and a corresponding parametric distribution, such as $\mathcal{G}\left(\alpha E_{n+1}, \beta\right)$, where $E_{n+1} > 0$ represents the exposure at year $n+1$, and $\alpha > 0$ and $\beta > 0$ are parameters to be fitted using the observations $(C_{1}^{i})_{1 \leq i \leq n}$, assuming the exposure information is available. Once this is done, we can combine the simulations of $C_{1}^{n+1}$ with the $(\widehat{f}_{j}^m)$ and $(\widehat{\sigma}_{j}^m)$, and then apply the bootstrap to obtain the simulations of $C_{n}^{n+1}$.
\end{remarque}

\section{Examples}

We implement our bootstrap method within a continuous-time framework, applying it to two data examples provided by  \cite{mack1993distribution}. The first dataset, initially introduced by \cite{taylor1983second}, is presented in Table \ref{triangle1}. The second dataset, pertaining to the mortgage guarantee business, is shown in Table \ref{triangle2}.

        \begin{table}[H]
        \begin{center}\setlength{\tabcolsep}{1.4mm}
        \fontsize{9pt}{12pt}\selectfont
        \begin{tabular}{|c|c c c c c c c c c c|}
  \hline
  $i$ & $C_{i,1}$ & $C_{i,2}$ & $C_{i,3}$ & $C_{i,4}$ & $C_{i,5}$ & $C_{i,6}$ & $C_{i,7}$ & $C_{i,8}$ & $C_{i,9}$ & $C_{i,10}$ \\
  \hline
  1 & 357848 & 1124788 & 1735330 & 2218270 & 2745596 & 3319994 & 3466336 & 3606286 & 3833515 & 3901463 \\
  2 & 352118 & 1236139 & 2170033 & 3353322 & 3799067 & 4120063 & 4647867 & 4914039 & 5339085 & \\
  3 & 290507 & 1292306 & 2218525 & 3235179 & 3985995 & 4132918 & 4628910 & 4909315 & & \\
  4 & 310608 & 1418858 & 2195047 & 3757447 & 4029929 & 4381982 & 4588268 & & & \\
  5 & 443160 & 1136350 & 2128333 & 2897821 & 3402672 & 3873311 & & & & \\
  6 & 396132 & 1333217 & 2180715 & 2985752 & 3691712 & & & & & \\
  7 & 440832 & 1288463 & 2419861 & 3483130 & & & & & & \\
  8 & 359480 & 1421128 & 2864494 & & & & & & & \\
  9 & 376686 & 1363294 & & & & & & & & \\
  10 & 344014 & & & & & & & & & \\
  \hline
\end{tabular}
\end{center}
\caption{The first dataset used in \cite{mack1993distribution} and originally presented by \cite{taylor1983second}.\label{triangle1}}
\end{table}

        \begin{table}[H]
        \begin{center}\setlength{\tabcolsep}{1.4mm}
        \fontsize{9pt}{12pt}\selectfont
        \begin{tabular}{|c|c c c c c c c c c|}
  \hline
  $i$ & $C_{i,1}$ & $C_{i,2}$ & $C_{i,3}$ & $C_{i,4}$ & $C_{i,5}$ & $C_{i,6}$ & $C_{i,7}$ & $C_{i,8}$ & $C_{i,9}$ \\
  \hline
  1 & 58046 & 127970 & 476599 & 1027692 & 1360489 & 1647310 & 1819179 & 1906852 & 1950105 \\
  2 & 24492 & 141767 & 984288 & 2142656 & 2961978 & 3683940 & 4048898 & 4115760 & \\
  3 & 32848 & 274682 & 1522637 & 3203427 & 4445927 & 5158781 & 5342585 & & \\
  4 & 21439 & 529828 & 2900301 & 4999019 & 6460112 & 6853904 & & & \\
  5 & 40397 & 763394 & 2920745 & 4989572 & 5648563  & & & & \\
  6 & 90748 & 951994 & 4210640 & 5866482 & & & & & \\
  7 & 62096 & 868480 & 1954797 & & & & & & \\
  8 & 24983 & 284441  & & & & & & & \\
  9 & 13121  & & & & & & & & \\
  \hline
\end{tabular}
\end{center}
\caption{The second dataset used in \cite{mack1993distribution}.\label{triangle2}}
\end{table}

We then compare our results with those obtained by Mack and the distribution derived from traditional bootstrap procedures.

\medbreak

We evaluate the conditional MSEP and the bootstrap distribution across the following models:
\begin{itemize}
    \item Mack's model \cite{mack1993distribution}, using the classic Mack's conditional MSEP formula and assuming a Log-normal parameterized distribution for the reserves' distribution;
    \item Mack's model with the bootstrap method;
    \item The time series model \cite{buchwalder2006mean} with the bootstrap technique;
    \item Our continuous-time model incorporating the bootstrap approach.

\end{itemize}
Let us briefly review the first three models. Our continuous-time model with bootstrap was described in the previous section. 

\subsection{Mack's model with a parameterized Log-normal distribution}

We employ the classic estimator $\widehat{C}_{i,n} := C_{i, n-i+1}\prod_{k=n-i+1}^{n-1}\widehat{F}_{k}$, which leads to the estimation of the expected value for the total reserve: 
    \[
     \widehat{\mu}_{R} := \widehat{R}.
    \]
We denote by $\widehat{\sigma}^{2}_{R}$ the conditional MSEP of \cite{mack1993distribution}. Finally, we approximate the distribution of the reserve with a Log-normal distribution by matching the moments:
    \[
        \mathcal{LN}\left(\log(\widehat{\mu}_{R}) - \frac12 \log\left(1+\frac{\widehat{\sigma}^{2}_{R}}{\widehat{\mu}_{R}}\right) , \log\left(1+\frac{\widehat{\sigma}^{2}_{R}}{\widehat{\mu}_{R}}\right)\right).
    \]

\subsection{Mack's model with Bootstrap}\label{M_bootstrap}

For comparison purposes, we calculate both the conditional MSEP and its distribution using a bootstrap method, where the conditional MSEP is the variance of the bootstrap distribution.

\medbreak

We employ the procedure outlined in \cite{england2006predictive}, which we briefly summarize here. First, we compute the Pearson residuals $(r_{i,j})$.

\medbreak

\noindent \textbf{1.} We simulate:
    \[
        C_{i,j+1}^m := \widehat{F}_{j}C_{i,j} + \widehat{\Sigma}_{j}\sqrt{C_{i,j}}r_{i,j}^m, \ \ 1 \leq m \leq M,
    \]
    where each $r_{i,j}^m$ is chosen from the $(r_{i,j})$ uniformly and independently. Using \eqref{est_fs}, we compute $(\widehat{F}_{j}^m, \widehat{\Sigma}_{j}^{m})_{1 \leq j \leq n-1}$ for $1 \leq m \leq M$. 

\medbreak

\noindent \textbf{2.}  We initiate the simulation with $C_{i,n-i+1}^{m} := C_{i,n-i+1}$, and then iteratively simulate the lower triangle for $2 \leq i \leq n$ as follows:
        \[
            C_{i,j+1}^m \sim \mathcal{N}\left(\widehat{F}_{j}^{m}C_{i,j}^m, (\widehat{\Sigma}_{j}^{m})^{2}C_{i,j}^m\right),
        \]
        and we deduce the bootstrap distribution of the total reserve with the formula (\ref{r_bootstrap}).

\subsection{Time series with Bootstrap}

The model developped in \cite{buchwalder2006mean} is based on the relation defined in (\ref{Cij_ts}), which is:

\begin{equation*}
    C_{i,j+1} = F_{j} C_{i,j} + \Sigma_{j} \sqrt{C_{i,j}} \varepsilon_{i,j},
\end{equation*}
where the $\varepsilon$'s are independent and centered with unit variance. We introduce the following hypothesis:
\[
	(\varepsilon_{i,j})_{1 \leq i, j \leq n} \overset{i.i.d.}{\sim} \mathcal{N}(0, 1).
\]
Now, we describe the bootstrap method for this model.

\medbreak

\noindent \textbf{1.} We simulate:
    \[
        C_{i,j+1}^m \sim \mathcal{N}\left(\widehat{F}_{j}C_{i,j}, \widehat{\Sigma}_{j}^{2}C_{i,j}\right), \ \ 1 \leq m \leq M.
    \]
Using \eqref{est_fs}, we derive $(\widehat{F}_{j}^m, \widehat{\Sigma}_{j}^m)_{1 \leq j \leq n-1}$ for $1 \leq m \leq M$. Additionally, it is noteworthy that when $C_{i,j}$ is fixed,

\begin{equation}
    \begin{aligned}
    \widehat{F}_{j}^m &\sim \mathcal{N}\left(\widehat{F}_{j}, \frac{\widehat{\Sigma}^2_j}{\sum_{i=1}^{n-j}C_{i,j}}\right), \\
    (\widehat{\Sigma}_{j}^{m})^2 &\sim \widehat{\Sigma}_j^2\frac{\chi^2_{n-j-1}}{n-j-1},
    \end{aligned}
\end{equation}
and both are independent. We can simulate the $(\widehat{F}_{j}^m, \widehat{\Sigma}_{j}^{m})_{1 \leq j \leq n-1}$ directly.

\medbreak

\noindent \textbf{2.} As in Section \ref{M_bootstrap}, we begin with $C_{i,n-i+1}^{m} := C_{i,n-i+1}$, we simulate iteratively, for $2 \leq i \leq n$ the lower triangle:
        \[
            C_{i,j+1}^m \sim \mathcal{N}\left(\widehat{F}_{j}^{m}C_{i,j}^m, (\widehat{\Sigma}_{j}^{m})^{2}C_{i,j}^m\right),
        \]
        and we deduce the bootstrap distribution of the total reserve with the formula \eqref{r_bootstrap}.

\subsection{Comparison and conclusion}

All bootstraps hereafter use $M = 10^7$ simulations. We begin by computing the conditional MSEP of the different methods using the first dataset from Table \ref{triangle1}, along with the $99.5\%$ quantile, all expressed relative to the common reserve estimator $\widehat{R}$. The conditional MSEP for Mack's Log-normal model is computed using the classical estimator from \cite{mack1993distribution}, while those for other methods are derived from the empirical bootstrap distributions of the reserves.

\begin{table}[H]
        \begin{center}\footnotesize
        \begin{tabular}{|c|c|c|}
  \hline
   Method & $ \sqrt{\widehat{MSEP}}$ (in \% of $\widehat{R}$) & $Q(R; 99.5\%) - \widehat{R}$ (in \% of $\widehat{R}$)  \\
  \hline
  Mack Log-normal & 13.0995 & 38.7466 \\
  Mack Bootstrap & 11.7585 & 33.0675 \\ 
  Time series Bootstrap & 13.1030 & 36.2963 \\
  Continuous-time Bootstrap & 13.1039 & 37.0219\\
  \hline
\end{tabular}
\end{center}
\caption{Conditional MSEP for the three other methods introduced and our \emph{continuous-time Bootstrap} from Section \ref{C_bootstrap} with the dataset from Table \ref{triangle1}.}\label{table_bootstrap}
\end{table}

Our results closely align with Mack's original formula for the conditional MSEP, as seen in the Time Series Bootstrap approach. The Mack Bootstrap yields a lower conditional MSEP, differing from the Time Series Bootstrap only in the simulations of $(\widehat{F}_{j}^m, \widehat{\Sigma}_{j}^{m})_{1 \leq j \leq n-1}$. This discrepancy primarily arises from the Pearson residuals being more \emph{regular}, indicating smaller values. Regarding quantiles, our analysis reveals a modest decrease relative to Mack Log-normal approach, a slight increase compared to the Time Series Bootstrap, and a substantial increase relative to Mack Bootstrap. For Mack's classic approach, which assumes a distribution-free framework, we adopted a log-normal distribution for the quantiles. This choice has notable implications: when a Gamma distribution is used instead, the quantile excess decreases to 36.95\%.

\medbreak

In our simulations, both the Mack Bootstrap and Time Series Bootstrap methods occasionally yield $C_{j}^{i} < 0$. Although rare in this example due to the data's regularity, occurring roughly once every $10^5$ simulations, such occurrences have been removed, with the introduced bias being negligible. The processes $(C_{t}^i)_{1 \leq t \leq n}$ remain non-negative in the continuous-time model, by construction. Nevertheless, in Remark \ref{PC0}, we noted that $\mathbb{P}(C_{j}^i = 0) > 0$ and asserted it to be numerically negligible. The highest probabilities arise for $j = 2$, and we have:
\[
	\mathbb{P}(C_{2}^n = 0) = \exp(-52.3031) \approx 1.9277 \times 10^{-23}.
\]

With less regularly structured data, the Mack or Times series Bootstrap methods might more frequently yield $C_{j}^{i} < 0$, potentially introducing bias if the corresponding simulations are removed or set to 0. However, such occurrences never arise within our continuous-time framework.

\medbreak

In Figure \ref{fig_bootstrap}, we display the complete distributions associated with the various models and our framework.

\begin{figure}[H]
\centering
\includegraphics[scale=0.595]{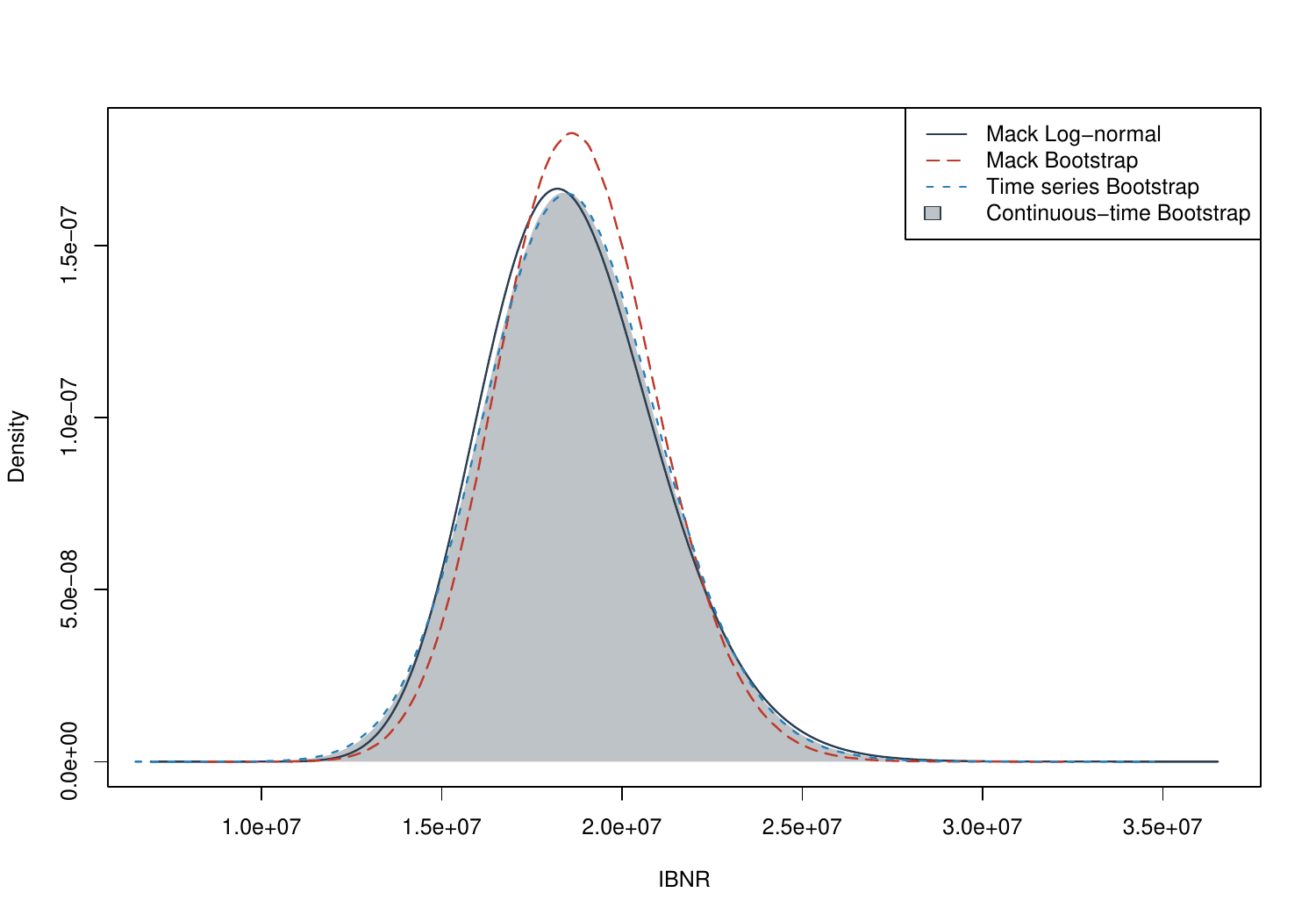}
\caption{Estimated conditional distributions of the total reserve with the dataset from Table \ref{triangle1}.\label{fig_bootstrap}}
\end{figure}

As observed in Table \ref{table_bootstrap}, the distribution within our framework closely resembles that of Mack with the Log-normal parameterized distribution and the Time Series Bootstrap. Our approach offers a significant advantage: it employs continuous simulation to eliminate negative values without introducing bias while maintaining the integrity of moment assumptions.

\medbreak

Finally, in Figure \ref{cn2}, we illustrate the bootstrap distribution of $C_{2}^{n}$ derived from the estimated $\widehat{F}_{1}$ and $\widehat{\Sigma}^2_{1}$, alongside comparisons with a Gaussian distribution and a Gamma distribution possessing equivalent moments.

\begin{figure}[H]
\centering
\includegraphics[scale=0.595]{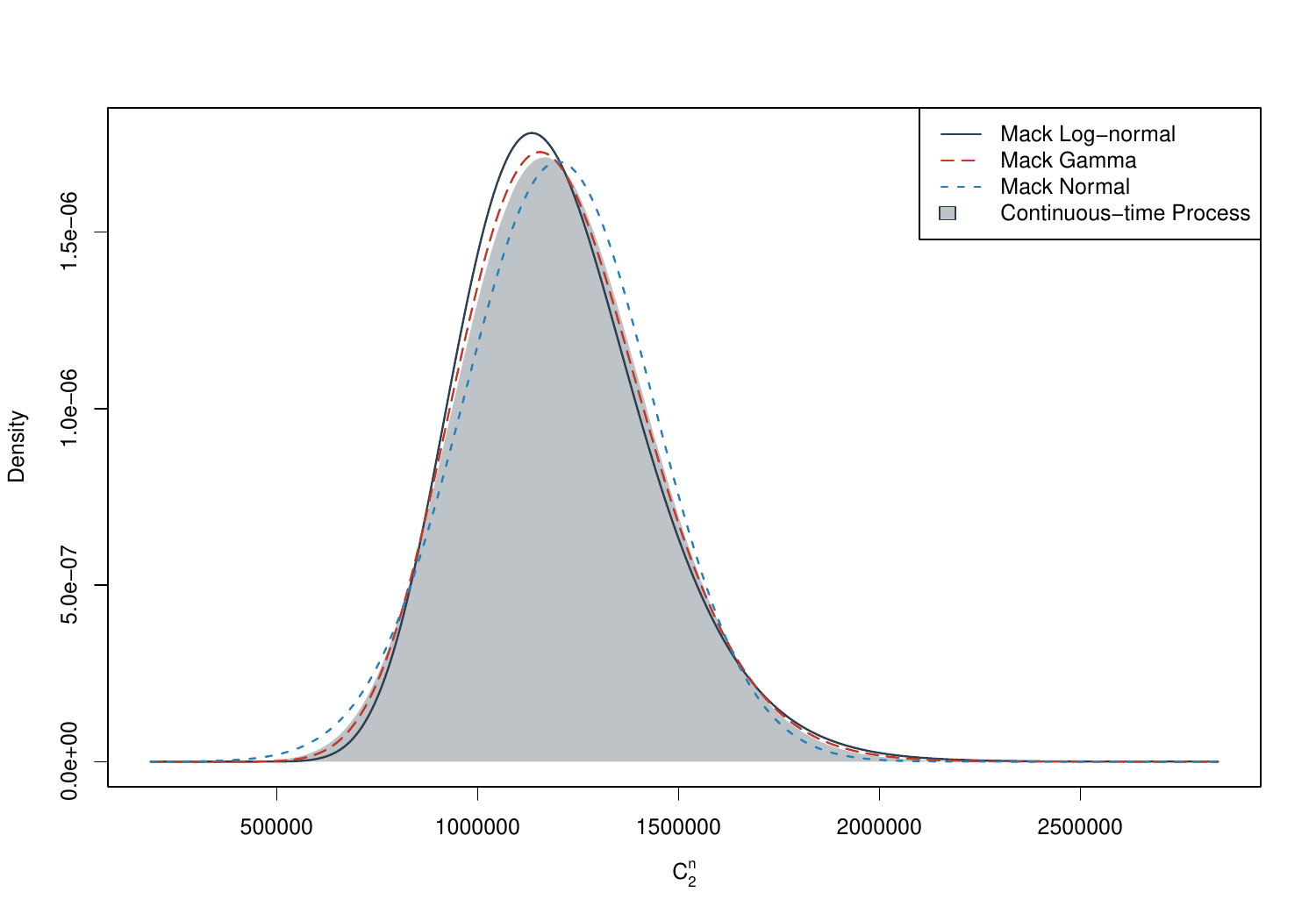}
\caption{Estimated distributions of $C_{2}^{n}$ in our continuous-time model, $\mathbb{P}(C_{2}^n = 0)$ is neglected.}\label{cn2}
\end{figure}

We observe that our continuous-time model provides a slightly asymmetric distribution, closely resembling that of a Gamma distribution.

\bigbreak

We now examine the second dataset introduced in Table \ref{triangle2}. This dataset exhibits less regularity, featuring a development factor of approximately 11 in the first year. As with the previous analysis, we compute the conditional MSEP for the various methods, this time using the new dataset, alongside the $99.5\%$ quantile, all expressed relative to the common reserve estimator $\widehat{R}$. Notably, the bootstrap simulations yield non-negligible negative values, with 18.9\% of Mack’s bootstrap simulations and 26.2\% of Time Series bootstrap simulations encountering such outcomes. Whenever $C_{j}^{i} < 0$, we replace these with 0, a choice we will discuss further at the end.

      \begin{table}[H]
        \begin{center}\footnotesize
        \begin{tabular}{|c|c|c|}
  \hline
   Method & $ \sqrt{\widehat{MSEP}}$ (in \% of $\widehat{R}$) & $Q(R; 99.5\%) - \widehat{R}$ (in \% of $\widehat{R}$)  \\
  \hline
  Mack Log-normal & 25.6337 & 85.5185 \\
  Mack Bootstrap & 22.9662 & 77.2303 \\ 
  Time series Bootstrap & 24.6414 & 76.9349 \\
  Continuous-time Bootstrap & 25.7493 & 88.3811 \\
  \hline
\end{tabular}
\end{center}
\caption{The conditional MSEP of the three other methods introduced and our continuous-time bootstrap from Section \ref{C_bootstrap} with the dataset from Table \ref{triangle2}.}\label{table_bootstrap2}
\end{table}

With this new dataset, the results diverge from previous findings. The Mack Bootstrap and Time Series Bootstrap yield lower conditional MSEP values, whereas the continuous-time bootstrap produces a conditional MSEP comparable to that of the classic Mack method. Similarly, quantile analysis reveals consistent patterns: the continuous-time bootstrap exhibits the highest quantile, surpassing the Mack Log-normal by 3\%. However, adjusting the latter to employ a Gamma distribution reduces the quantile to $78.2503\%$, highlighting the critical influence of the chosen distribution.

\medbreak

In Figure \ref{fig_bootstrap2}, we display the complete distributions associated with the various models and our framework.

\begin{figure}[H]
\centering
\includegraphics[scale=0.595]{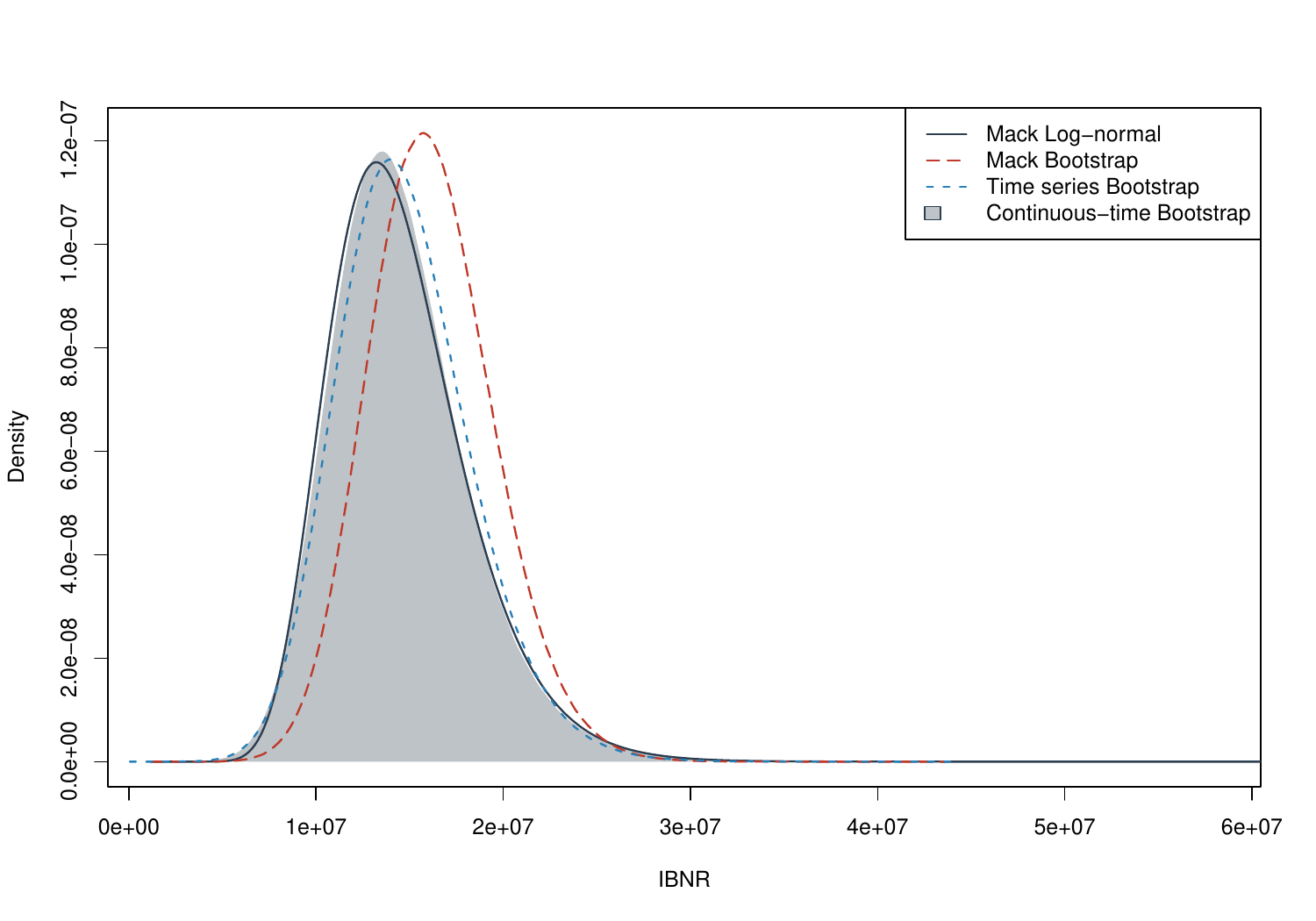}
\caption{Estimated conditional distributions of the total reserve with the dataset from Table \ref{triangle2}.\label{fig_bootstrap2}}
\end{figure}

This time, the empirical distributions exhibit clear differences.

\medbreak

For this dataset, across all bootstrap methods, $20\%$ to $40\%$ of the simulations produced at least one $C_{j}^{i} < 0$, which we replaced with 0. In the continuous-time bootstrap, such values were inherently zero, requiring no adjustment. This correction may introduce bias into the Mack Bootstrap and Time Series Bootstrap procedures. For the continuous-time bootstrap, we calculate:
\[
	\mathbb{P}(C_{2}^n = 0) = \exp(-1.8102) \approx 0.1636.
\]
This probability is no longer negligible. This is perfectly taken into account within the simulations and does not bias the methodology; it preserves the moments and remains tailored to the case. However, a probability of $\mathbb{P}(C_{2}^n = 0) \approx 0.1636$ seems unrealistic in practice. With a highly irregular triangle, we approach the limits of the method. The small value of $C_{1}^{n}$ plays a significant role here, being notably low relative to its column, especially in combination with the irregular triangle. If we replace with $C_{1}^{n-1}$ (24983 instead of 13121), we obtain $\mathbb{P}(C_{2}^{n-1} = 0) \approx 0.03184$, which is more reasonable but still not negligible.

\FloatBarrier

\section*{Acknowledgments}

The author acknowledges the financial support provided by the \emph{Fondation Natixis}.

\bibliographystyle{plain}
\bibliography{bibliographie}

\begin{thebibliography}{10}

\bibitem{bischofberger2020continuous}
Stephan~M Bischofberger, Munir Hiabu, and Alex Isakson.
\newblock Continuous chain-ladder with paid data.
\newblock {\em Scandinavian Actuarial Journal}, 2020(6):477--502, 2020.

\bibitem{buchwalder2006mean}
Markus Buchwalder, Hans B{\"u}hlmann, Michael Merz, and Mario~V W{\"u}thrich.
\newblock The mean square error of prediction in the chain ladder reserving
  method (mack and murphy revisited).
\newblock {\em ASTIN Bulletin: The Journal of the IAA}, 36(2):521--542, 2006.

\bibitem{dawson2017introductory}
Donald~A Dawson.
\newblock Introductory lectures on stochastic population systems.
\newblock {\em arXiv preprint arXiv:1705.03781}, 2017.

\bibitem{england2006predictive}
Peter~D England and Richard~J Verrall.
\newblock Predictive distributions of outstanding liabilities in general
  insurance.
\newblock {\em Annals of Actuarial Science}, 1(2):221--270, 2006.

\bibitem{feller1971introduction}
William Feller.
\newblock An introduction to probability theory and its applications.
\newblock {\em John Wiley}, 2, 1957.

\bibitem{ikeda2014stochastic}
Nobuyuki Ikeda and Shinzo Watanabe.
\newblock {\em Stochastic differential equations and diffusion processes},
  volume~24.
\newblock Elsevier, 2014.

\bibitem{kremer1985einfuhrung}
Erhard Kremer.
\newblock Einf{\"u}hrung in die versicherungsmathematik.
\newblock {\em Vandenhoek \& Ruprecht}, 1985.

\bibitem{kuang2009chain}
Di~Kuang, Bent Nielsen, and Jens~Perch Nielsen.
\newblock Chain-ladder as maximum likelihood revisited.
\newblock {\em Annals of Actuarial Science}, 4(1):105--121, 2009.

\bibitem{lamperti1967continuous}
J.~Lamperti.
\newblock Continuous-state branching processes.
\newblock {\em Bulletin of the American Mathematical Society}, 73:382--386,
  1967.

\bibitem{mack1991simple}
Thomas Mack.
\newblock A simple parametric model for rating automobile insurance or
  estimating ibnr claims reserves.
\newblock {\em ASTIN Bulletin: The Journal of the IAA}, 21(1):93--109, 1991.

\bibitem{mack1993distribution}
Thomas Mack.
\newblock Distribution-free calculation of the standard error of chain ladder
  reserve estimates.
\newblock {\em ASTIN Bulletin: The Journal of the IAA}, 23(2):213--225, 1993.

\bibitem{mack2006mean}
Thomas Mack, Gerhard Quarg, and Christian Braun.
\newblock The mean square error of prediction in the chain ladder reserving
  method--a comment.
\newblock {\em ASTIN Bulletin: The Journal of the IAA}, 36(2):543--552, 2006.

\bibitem{meleard2016modeles}
Sylvie M{\'e}l{\'e}ard.
\newblock {\em Mod{\`e}les al{\'e}atoires en Ecologie et Evolution}.
\newblock Springer, 2016.

\bibitem{miranda2013continuous}
Mar{\'\i}a Dolores~Mart{\'\i}nez Miranda, Jens~Perch Nielsen, Stefan Sperlich,
  and Richard Verrall.
\newblock Continuous chain ladder: Reformulating and generalizing a classical
  insurance problem.
\newblock {\em Expert Systems with Applications}, 40(14):5588--5603, 2013.

\bibitem{revuz2013continuous}
Daniel Revuz and Marc Yor.
\newblock {\em Continuous martingales and Brownian motion}, volume 293.
\newblock Springer Science \& Business Media, 2013.

\bibitem{taylor1983second}
Greg~C Taylor and Frank~R Ashe.
\newblock Second moments of estimates of outstanding claims.
\newblock {\em Journal of Econometrics}, 23(1):37--61, 1983.

\bibitem{yamada1971uniqueness}
Toshio Yamada and Shinzo Watanabe.
\newblock On the uniqueness of solutions of stochastic differential equations.
\newblock {\em Journal of Mathematics of Kyoto University}, 11(1):155--167,
  1971.

\end{thebibliography}

\end{document}